\documentclass[aps,twocolumn,nofootinbib,longbibliography]{revtex4-1}

\usepackage{amsmath,amsthm,amssymb}
\usepackage{enumerate}
\usepackage{ctable,colortbl}
\usepackage[english]{babel}
\usepackage{graphicx}
\usepackage{hyperref}
\usepackage{xcolor}
\hypersetup{colorlinks=true,citecolor=blue,linkcolor=blue,filecolor=blue,urlcolor=blue}
\usepackage{mdframed}
\colorlet{shadecolor}{orange!15}
\usepackage{setspace}
\usepackage{nicefrac}
\usepackage{lmodern}

\theoremstyle{plain}
\newtheorem{theorem}{Theorem}
\newtheorem{lemma}[theorem]{Lemma}

\newtheorem{corollary}[theorem]{Corollary}

\newcommand*{\cB}{\mathcal{B}}

\newcommand*{\cF}{\mathcal{F}}

\newcommand*{\cH}{\mathcal{H}}

\newcommand*{\cK}{\mathcal{K}}

\newcommand*{\cX}{\mathcal{X}}

\newcommand*{\cZ}{\mathcal{Z}}

\newcommand*{\id}{\mathrm{id}}

\newcommand*{\ket}[1]{{| #1 \rangle}}
\newcommand*{\bra}[1]{{\langle #1 |}}
\newcommand*{\spr}[2]{\langle #1 | #2 \rangle}

\colorlet{shadecolor}{orange!15}
  
   \newmdenv[ %
 linewidth = 0pt, %
  leftmargin = -5, %
  rightmargin = -5, %
  backgroundcolor = orange!15, % 
  innertopmargin = 4, %
   innerbottommargin = 4, %
  splittopskip = \topskip, %
  footnoteinside = true, %
  ]{emphbox}

\definecolor{orange}{rgb}{1,0.5,0}
\definecolor{darkgreen}{rgb}{0,0.4,0}

\newcommand*{\hor}{\mathrm{h}}
\newcommand*{\ver}{\mathrm{v}}
\newcommand*{\diag}{\mathrm{d}}

\newcommand*{\Exp}{\mathsf{Exp}}
\newcommand*{\ExpMeas}{\mathbf{P\&M}}
\newcommand*{\ExpMeasRep}{\mathbf{P\&M}^*}

\newcommand*{\BornBayes}{\textsf{\textit{BornB}}}
\newcommand*{\BornFreq}{\textsf{\textit{BornF}}}
\newcommand*{\BornDet}{\textsf{\textit{Overlap}}}
\newcommand*{\Robust}{\textsf{\textit{Robust}}}
\newcommand*{\Repeat}{\textsf{\textit{Repeat}}}
\newcommand*{\Symmetry}{\textsf{\textit{Symmetry}}}

\newcommand*{\texts}[1]{{\begin{minipage}{7.8cm} \vspace{0em}   \raggedright \begin{spacing}{0.88}  \hangindent=0.7em {\footnotesize \hangafter=1  ``#1''  } \end{spacing} \vspace{0.25em} \end{minipage}} }

\newcommand*{\textss}[1]{{\begin{minipage}{6cm} \vspace{0.4em}   \raggedright \begin{spacing}{0.88}  \hangindent=0.7em {\footnotesize \hangafter=1  ``#1''  } \end{spacing} \vspace{0.4em} \end{minipage}} }

\newcommand*{\zspec}{{\hat{z}}}

\begin{document}

\title{A non-probabilistic substitute for the Born rule}

\author{Daniela Frauchiger}
\author{Renato Renner}\email{renner@ethz.ch}
\affiliation{Institute for Theoretical Physics, ETH Zurich, Switzerland}

\begin{abstract}
The Born rule assigns a probability to any possible outcome of a quantum measurement, but leaves open the question how these probabilities are to be interpreted and, in particular, how they relate to the outcome observed in an actual experiment. We propose to avoid this question by replacing the Born rule with two non-probabilistic postulates: (i)~the projector associated to the observed outcome must have a positive overlap with the state of the measured system; (ii)~statements about observed outcomes are robust, that is, remain valid under small perturbations of the state. We  show that the two postulates suffice to retrieve the interpretations of the Born rule that are commonly used for analysing experimental data. 
\end{abstract}

\maketitle

\section{Introduction}

In standard quantum mechanics the Born rule~\cite{Born26} has the status of a postulate. It builds upon two other postulates of the theory, namely that the states of a quantum system  can be represented by normalised vectors $\psi$ in a Hilbert space $\cH$ (which may include auxiliary systems) and that measurements correspond to families $\{\pi_z\}_{z \in \cZ}$ of projectors that sum up to the identity on $\cH$.\footnote{Alternatively, one may use density operators to specify the system's states. However, as the Born rule plays a significant role for justifying the density operator formalism, we omit it here in order to avoid any possible circularity in our arguments.}  The aim of the Born rule is to relate these purely mathematical concepts to experimental observations.\footnote{In other physical theories, such as classical mechanics, the relation between the mathematical formalism and experimental observations is more intuitive and not usually specified explicitly.} 
It asserts that if a system in state $\psi$ is subject to a measurement $\{\pi_z\}_{z \in \cZ}$ then outcome $z \in \cZ$ is observed with probability
\begin{equation}\label{eq_tradBorn}
  P_{\psi}(z)= \| \pi_z \psi\|^2  \ . \tag{$\star$}
\end{equation}
This expression is sometimes termed the  ``probabilistic axiom'' of quantum theory. However, unless the notion of  probabilities is given a physical interpretation, \eqref{eq_tradBorn} remains an ambiguous |  if not an empty | statement~\cite{Appleby05}.  In other words, it is necessary to specify  in what sense the real-valued number $P_{\psi}(z)$ assigned to~$z$ relates to the actual observation of $z$ (cf.\ Fig.~\ref{fig_halfway}). 

There exists a plethora of literature on how to give meaning to the mathematical notion of probabilities, and hence how to understand~\eqref{eq_tradBorn} and use it to analyse experimental data.  Here we consider the two most common interpretations of probabilities, \emph{frequentist probabilities} and \emph{Bayesian probabilities}. Each of them leads to a specific reading of \eqref{eq_tradBorn}, which, for later reference, we term \BornFreq{} and \BornBayes{}, respectively. We will make these precise in Sections~\ref{sec_Frequentist} and~\ref{sec_Bayes}, based on a framework to be introduced in Section~\ref{sec_framework}.

\begin{itemize}
\item \BornFreq{} relates \eqref{eq_tradBorn} to an experiment where the measurement  is repeated many times under identical initial conditions. $P_{\psi}(z)$ is interpreted as the \emph{frequency} of occurrences of $z$ in the sequence of outcomes, in the limit of infinitely many repetitions.
\item \BornBayes{} uses \eqref{eq_tradBorn}  to characterise a \emph{subjective belief} about  the outcome of a future measurement.  $P_{\psi}(z)$  is interpreted as the maximum price a rational agent would pay to enter a bet with payoff~$1$ if outcome~$z$ is observed, and no payoff otherwise.   
\end{itemize}

\begin{figure}[t]
\centering
\includegraphics[trim= 4.4cm 9.1cm 4.4cm 11.0cm, clip=true, width=0.48\textwidth]{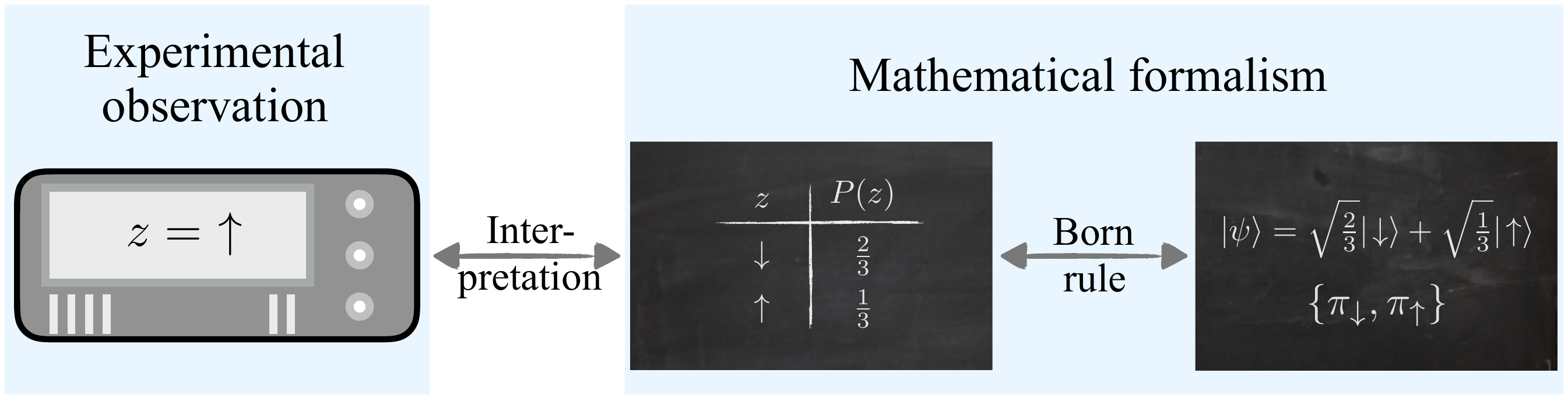}
\caption{\label{fig_halfway}  \emph{The need of an interpretation.} The Born rule~\eqref{eq_tradBorn} relates amplitudes to probabilities, both of them being purely mathematical notions. To link them to actual observations, the probabilities need to be equipped with an interpretation.}
\end{figure}

If one wishes to regard the Born rule as a basic postulate of physics then neither of these readings is entirely satisfactory. \BornFreq{} is constrained to hypothetical situations where a measurement can be repeated arbitrarily many times under identical conditions. \BornBayes{} does not suffer from this restriction, but talks about an agent's belief rather than about physical observations. In addition, it is also questionable whether a particular quantitative expression such as~\eqref{eq_tradBorn} should really be given the status of a postulate | instead of being derived from more fundamental principles. The principle of relativity, for instance, on which the theory of relativity is based, does not refer to any particular numerical values. 

This raises the question whether the Born rule could be substituted by something else. Here we propose a pair of alternative postulates.  Crucially, neither of them relies on the notion of probabilities  (see Fig.~\ref{fig_full})!
\begin{itemize}
  \item \BornDet{}: If $\psi$ has no overlap with $\pi_z$ then the corresponding  outcome $z$ is not observed.
  \item \Robust{}: Statements that relate $\psi$ to observed outcomes are robust under small perturbations of~$\psi$.
\end{itemize}
While \BornDet{}  may be interpreted as a special case of~\eqref{eq_tradBorn}, we stress that it is of a different kind: rather than assigning probabilities to quantum states, it relates quantum states directly to observations. This, together with \Robust{}, also makes it falsifiable by experiments. 

\begin{figure}[t]
\centering
\includegraphics[trim=4.4cm 9.1cm 4.4cm 12.4cm, clip=true, width=0.48\textwidth]{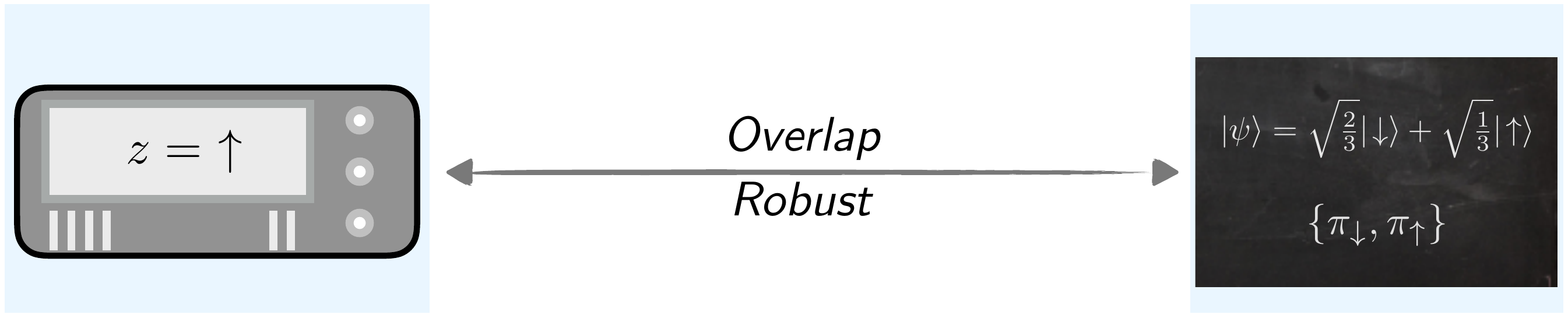}

\caption{\label{fig_full} \emph{Non-probabilistic postulates.} The postulates proposed in this work, \BornDet{} and \Robust{}, establish a direct connection between quantum states and observed measurement outcomes, without a detour via probabilities.}
\end{figure}

The main technical contribution of this work is to show that, if one takes \BornDet{} and \Robust{} as postulates, the above interpretations of the Born rule can be derived. The frequentist interpretation, in particular, is obtained directly without any additional assumptions, i.e.,
\begin{align} \label{eq_claimBornFreq}
  \text{\BornDet} \wedge   \text{\Robust} \implies   \text{\BornFreq}  \ .
\end{align}
The proof of this claim is provided in Section~\ref{sec_Frequentist}.

To retrieve the  subjective interpretation, i.e., a statement about an agent's personal belief,  we need to make some extra assumptions concerning the agent's reasoning.
\begin{itemize}
  \item  \Repeat{}: The agent's belief about the outcome of a prepare-and-measure experiment does not depend on whether or not they plan to repeat the same experiment later on.
    \item \Symmetry{}: The agent's belief about a sequence of outcomes obtained by repeating the same prepare-and-measure experiment does not depend on how the sequence is ordered.
  \end{itemize}
Combined with the above, these assumptions  suffice to retrieve the Bayesian interpretation of the Born rule, i.e., 
\begin{align} \label{eq_claimBornBayes}
  \left. \begin{matrix}
    \text{\BornDet} \wedge \text{\Robust} \\
    \text{\Repeat} \wedge \text{\Symmetry} 
  \end{matrix}
  \right\}
  \implies \text{\BornBayes} \ .
\end{align}
The derivation of this claim will be based on an idea of de Finetti~\cite{deFin37} and is explained in Section~\ref{sec_Bayes}.

\section{Related work} 

Various approaches to derive the Born rule from more fundamental principles have been proposed in the literature. One of them is to use purely mathematical axioms, according to which probabilities are  abstract values without  any specific meaning. A prominent example for such an approach is Gleason's theorem~\cite{Gleason57}. It is based on a non-contextuality assumption: given a fixed quantum state $\psi \in \cH$,  the probability $P_{\psi}(z)$ assigned to any possible outcome~$z$  must be uniquely determined by the corresponding projector $\pi_z$ (and hence be independent of the other projectors that define the measurement). In addition, $\cH$ must be at least $3$-dimensional.  The theorem then asserts that \eqref{eq_tradBorn} holds.\footnote{More precisely, it asserts that $P(\cdot)$ is a convex combination of functions of the form~\eqref{eq_tradBorn}, for different choices of $\psi \in \cH$.}  

More recently, Saunders~\cite{Sau04} as well as  Auff{\`e}ves and Grangier~\cite{AufGra15b} proposed modifications of Gleason's argument, which use  more operationally motivated axioms.  Zurek demonstrated that \eqref{eq_tradBorn}   can also be obtained from a symmetry principle, called ``environment-assisted invariance''~\cite{Zur05}, together with certain assumptions about probabilities, made explicit in~\cite{SchlossFine05}. A similar derivation of the probabilistic Born rule has been proposed by Lesovik~\cite{Lesovik14}, who uses a symmetry assumption together with the assumption that the probability of finding a particle in a given region is determined by its wave function amplitude in this region.  The common feature of these arguments, as well as of Gleason's theorem, is that they are based on assumptions about the probability distribution of measurement outcomes. In that respect they are somewhat orthogonal to our objective, which is to have postulates that talk about actual outcomes rather than (abstract) probabilities.  

An approach to obtain the Born rule  without resorting to probabilistic axioms is the decision-theoretic argument by Deutsch~\cite{Deutsch99}, and later refined by Wallace~\cite{Wal09}.  They showed that under certain assumptions about rationality, if a system in state~$\psi$ is measured, a rational agent will bet on the outcome $z$ with the maximal value $P_{\psi}(z)$ as given by~\eqref{eq_tradBorn}. The statement they derive is thus similar in spirit to the  Bayesian interpretation of the Born rule, \BornBayes{}. The original claim that the axioms used in~\cite{Deutsch99,Wal09} are non-probabilistic has however been questioned by Barnum \emph{et al.}~\cite{Barnum00}, who remarked that the argument implicitly uses an assumption that relates probabilities to the indistinguishability of particular events.
% (see~\cite{Baker07} for another criticism).
 
A rather different line of reasoning, based on typicality arguments, was proposed by Everett~\cite{Everett57}, Finkelstein~\cite{Finkelstein63}, DeWitt and Graham~\cite{DeWitt70,Graham73}, and Hartle~\cite{Hartle68}, and was later strengthened by Farhi, Goldstone and Gutmann~\cite{Far89}. The idea is to start from an axiom similar to~\BornDet{} and apply it to an experiment that is repeated infinitely often. It is then shown that \BornFreq{} holds for ``typical'' sequences of measurement outcomes, i.e., for a set of sequences of weight~$1$ according to some probability measure on the set of infinite sequences. However, as pointed out by Caves and Schack~\cite{Cav05}, additional assumptions are needed in order to define this probability measure. They, as well as Cassinello and S\'anchez-G\'omez~\cite{Cassinello96}, argue more generally that an axiom like \BornDet{} alone cannot suffice to retrieve \BornFreq{}.

A possible additional assumption that could be used to complete  such typicality-based arguments was proposed by Buniy, Hsu and Zee~\cite{Buniy06}. Their idea is to postulate that the quantum state space is fundamentally discrete. This postulate is related to our axiom~\Robust{}, in the sense that a theory with discrete state space is by definition robust against small perturbations of the states.  A similar robustness requirement has also been considered by de Raedt, Katsnelson and Michielsen, and, together with  assumptions about logical inference, shown to imply certain results of quantum theory~\cite{DeRaedt14}.

Recently various approaches have been put forward to derive the state-space structure of quantum theory from physical axioms. They are however mostly based on the idea that probabilities are an irreducible physical concept (see, e.g., \cite{Hardy01,Mueller12,Dariano06}). These works are hence complementary to ours: while we take the Hilbert space formalism of quantum theory as given, we do not presuppose a physical notion of probabilities.

\section{Physics without probabilities} \label{sec_framework}
 
\begin{figure}[t]
\includegraphics[trim= 2.2cm 2.68cm 2.5cm 0.8cm, clip=true, width=0.48\textwidth]{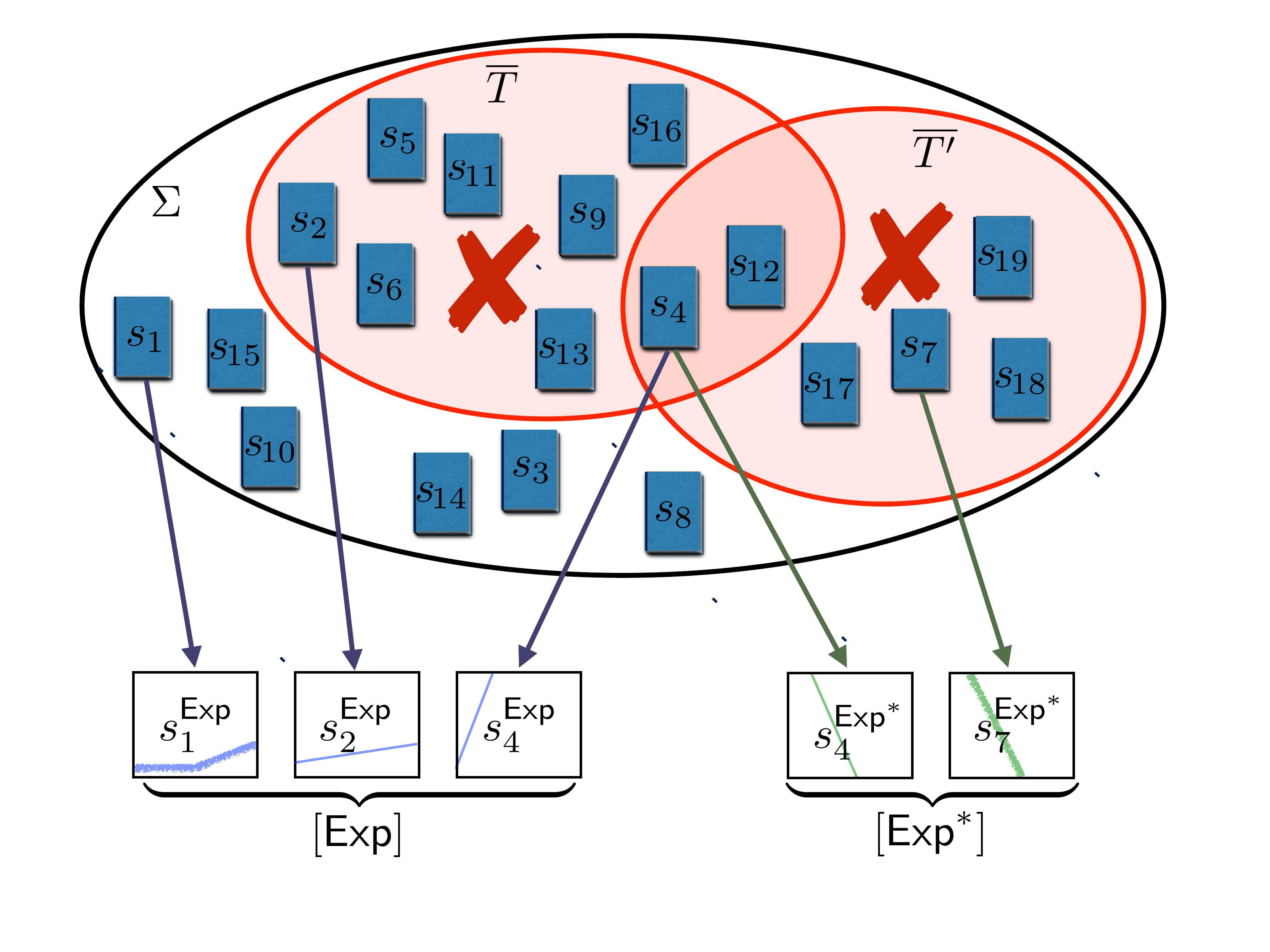}\caption{\emph{Structure of the framework.}   Physical theories are regarded as constraints on the set $\Sigma$ of possible stories. For example, $T$ could be quantum mechanics, and $T'$ special relativity. A story about a moving train that has an infinitely accurate location and velocity would be forbidden by~$T$, and hence be contained in the set~$\overline{T}$. A story according to which the train moves faster than the speed of light would be in $\overline{T'}$. The semantics of the stories is captured by their plots. A story~$s$ may specify plots for multiple experiments, e.g., $s^{\Exp} \! \subset \! [\Exp]$ and $s^{\Exp^*} \! \!  \! \subset \! [\Exp^*]$, for~$\Exp$ and~$\Exp^*$, respectively.  \label{fig_forbiddenallowed}
}
\end{figure}
 
 To formulate our claims we use a framework, introduced in~\cite{Fra16} (see also~\cite{FraRen16}), that allows us to reason about physical laws without relying on an \emph{a priori} notion of probabilities.\footnote{Most other frameworks, such as the Generalised Probabilistic Theories approach~\cite{Barrett07}, use conditional probability distributions to represent physical states.}  The approach is inspired by Deutsch, who maintains that the primary purpose of a physical theory is not to make predictions, but rather to ``tell stories'' that help us understand them~\cite{Deutsch97,Popper}. 

{\setlength{\abovedisplayskip}{5pt}
\setlength{\belowdisplayskip}{5pt}
Following this idea, a first basic ingredient to the framework we use is the notion of \emph{stories}.  Intuitively, they are  descriptions or explanations of a physical phenomenon. Technically,  we require that the set of all stories, which we denote by $\Sigma$, is countable (see the discussion at the end of Section~\ref{sec_Frequentist} for why this is important). One may think of them as (arbitrarily long) finite bit strings or, alternatively, finite sequences of English words.\footnote{Another possibility is to represent stories as sequences of ``clips'', a notion introduced in~\cite{Briegel12}.} An example would be 
\begin{align} \label{eq_trainstory}
  s_v = \textss{A train passes position $x=0$ at time $t=0$, and moves with velocity~$v$.} 
\end{align}
(where $v$ could be any value, e.g., $v=50 \,  \mathrm{m/sec}$).  A \emph{physical theory}, $T$,  provides criteria that \emph{forbid} certain stories. For example, any story $s_v$ with $v$ larger than the speed of light would be forbidden by special relativity. In the following, we denote by $\overline{T}$ the subset of~$\Sigma$ consisting of all stories that are forbidden by~$T$ (cf.\ Fig.~\ref{fig_forbiddenallowed}). 
}

The second basic ingredient to our framework is the notion of an \emph{experiment}. By definition, any experiment, $\Exp$, has an \emph{event space}, denoted by~$[\Exp]$. One should think of the elements of $[\Exp]$ as the events (observable or not) that are relevant to the analysis of the experiment. The event space~$[\Exp]$ is thus not dictated by nature, but rather can be chosen depending on the question one wishes to study.  For example, if one is interested in the kinematics of a train moving along a straight rail, one may consider an experiment $\Exp$ whose  event space $[\Exp]$ consists of pairs $(t, x) \in \mathbb{R}^2$, indicating the time and the position of the train (see Fig.~\ref{fig_Plot}).

\begin{figure}[t]
\includegraphics[trim= 2.5cm 15cm 2.5cm 2cm, clip=true, width=0.48\textwidth]{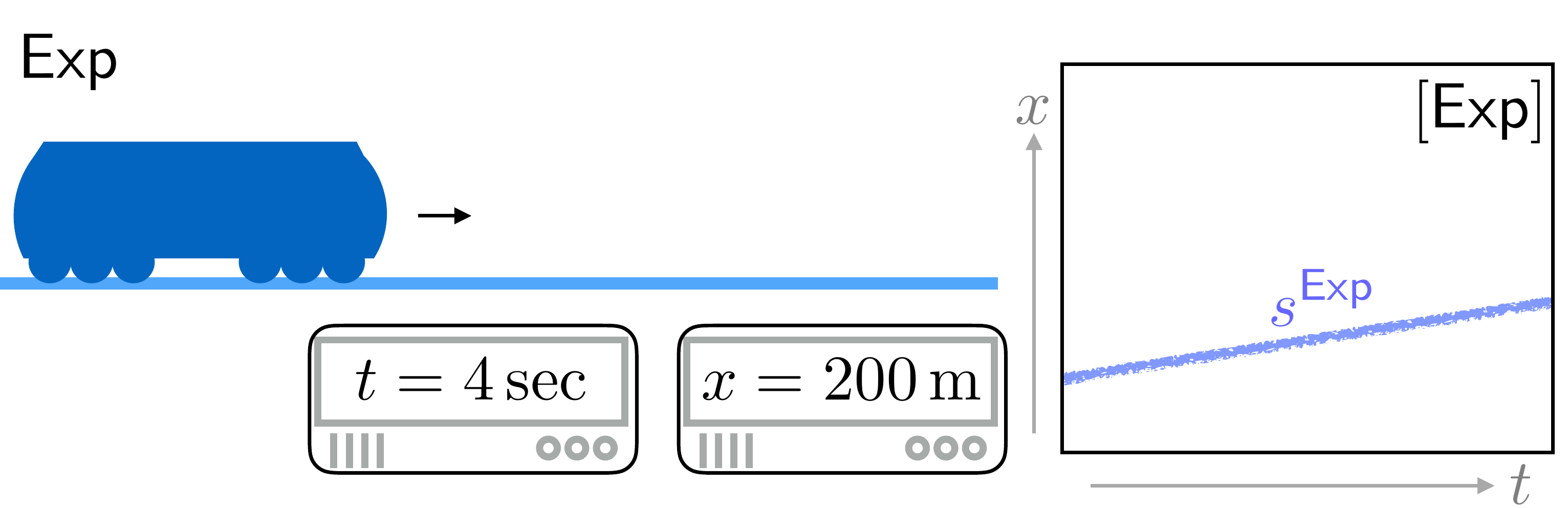}\caption{\emph{Event space and plot.}  Any experiment $\Exp$ has an event space $[\Exp]$. In the example experiment described in the text, the events are characterised by pairs $(t, x)$. A story  about~$\Exp$, such as $s_v$ (Eq.~\ref{eq_trainstory}), specifies a plot $s_v^{\Exp} \subset [\Exp]$ (Eq.~\ref{eq_trainstoryplot}), i.e., a subset of events that occur according to~$s_v$.
\label{fig_Plot}
}
\end{figure}

Given an experiment~$\Exp$, a story $s \in \Sigma$ may define a \emph{plot}, denoted by $s^{\Exp}$. This is a subset of $[\Exp]$, whose elements are to be interpreted as the events that \emph{actually} occur according to the story. For example, the plot of story~$s_v$ about~$\Exp$ could be taken to be 
\begin{align} \label{eq_trainstoryplot}
  s_v^{\Exp} = \{(t, {x = v t}): \, t \in \mathbb{R}\} \ .
\end{align}
Hence, for any experiment~$\Exp$, we have a function ${s \to s^{\Exp}}$ that equips stories~$s \in \Sigma$ with semantics, detailed in terms of plots. This function can be (and usually is) partial. That is,  one defines $s^{\Exp}$ only for particular stories~$s$ | those that tell us something about~$\Exp$ and are considered precise enough to specify a plot. 

Within this framework, laws of physics are expressed formally as conditions on the plots for appropriately chosen experiments. To talk about the laws of special relativity theory, for instance, one could employ experiments $\Exp$ with moving bodies such as the train example above. That the theory rules out velocities faster than the speed of light~$c$ could then be phrased as a condition on stories~$s$: it  should forbid any $s \in \Sigma$ that violates  
\begin{align*}
   \forall (t, x), (t', x') \in s^{\Exp} : \, (x'-x) < c(t'-t)  \ .
\end{align*}
 
\begin{figure}[t]
\centering
\includegraphics[trim= 1cm  4.1cm 1.1cm 6.8cm, clip=true, width=0.48\textwidth]{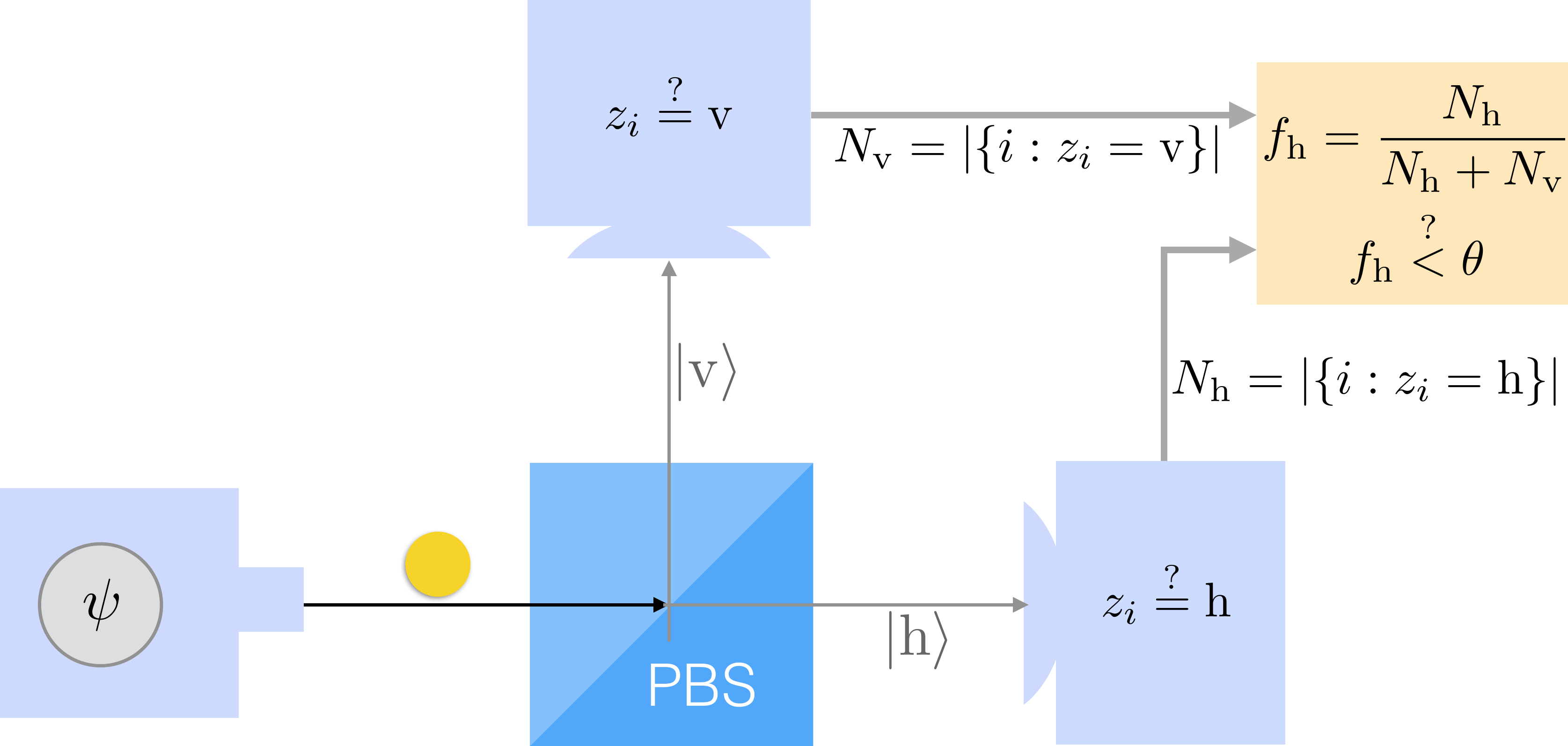}
\caption{\emph{Typical prepare-and-measure experiment.} A source emits photons with polarisation~$\psi$ towards a polarising beam splitter (PBS) that directs horizontally polarised photons to one detector and vertically polarised ones to another. The detection counts~$N_{\hor}$ and $N_{\ver}$ may be fed into an analyser that tests whether the relative frequencies are within a given range.
\label{fig_ExpLight}
}
\end{figure}

{\setlength{\abovedisplayskip}{5pt}
\setlength{\belowdisplayskip}{3pt}
For our study of the Born rule, we will consider general prepare-and-measure experiments. A typical example could consist of a source that emits polarised photons, which are subsequently measured to distinguish vertical and horizontal polarisations (Fig.~\ref{fig_ExpLight}). Denoting the polarisation state space of each photon by $\cK$,  the overall state space of all emitted photons can be written as 
\begin{align} \label{eq_Fock}
   \cH = \bigoplus_{n \in \mathbb{N}_0} \cK^{\otimes n} \ .
\end{align}
A natural choice for the event space of this experiment is then}
\begin{align} \label{eq_PMEventspace}
  [\Exp] = \bigl\{ (\Psi, z) :  \, \Psi  \in \cH, \,  z\!=\!(z_1, \ldots, z_n) \in \{\hor, \ver\}^{*} \bigr\} 
\end{align}
where $\{\hor, \ver\}^{*}$ denotes the set of  outcome tuples $z=(z_1, \ldots, z_n)$   of arbitrary (but finite) length~$n$, with $z_i \in \{\hor, \ver\}$ indicating whether the $i$th measurement gave horizontal or vertical polarisation. A few examples of stories one can tell about this experiment are listed in Table~\ref{tab_StoriesMeas}.  According to the standard understanding of the Born rule, stories $s_2$, $s_3$, and~$s_6$ should be ruled out, whereas the situation is less obvious for~$s_4$.

\section{Retrieving the Frequentist Rule} \label{sec_Frequentist}

The aim of this section is to make precise and prove claim~\eqref{eq_claimBornFreq}. We start by providing definitions for the expressions \BornDet{}, \Robust{}, and \BornFreq{}. To formulate them, we take for granted that states of a quantum system can be represented by vectors in a Hilbert space~$\cH$, and that measurements on that system correspond to families of projectors $\{\pi_z\}_{z \in \cZ}$ on $\cH$ such that $\sum_{z \in \cZ} \pi_z = \id_{\cH}$. (This includes situations where the system of interest is entangled with its environment, or where the measurement is not projective, in which case $\cH$ must be taken to be the joint Hilbert space of the system together with parts of the environment.) We always assume that $\cH$ is separable and that $\cZ$ is countable. 

We denote by~$\ExpMeas_{\cH, \{\pi_z\}}$ the set of all prepare-and-measure experiments in which a system is prepared in a state~$\psi \in \cH$ and measured with respect to $\{\pi_z\}_{z \in \cZ}$. Following the approach introduced in Section~\ref{sec_framework}, we may assign to them the event space
\begin{align*}
  [\Exp] = \{(\psi, z) : \, \psi \in \cH, \, z \in \cZ\} \ .
\end{align*}
Note that the experiment described around~\eqref{eq_PMEventspace} can be regarded as one of $\ExpMeas_{\cH, \{\pi_z\}}$  | just identify $z$ with the tuple $(z_1, \ldots, z_n)$ of individual measurement outcomes, and  $\psi$  with the joint state $\Psi$ of all measured subsystems.

 \begin{figure}[t]
 \begin{center}\includegraphics[trim= 2.5cm 16cm 2.5cm 2.7cm, clip=true, width=0.48\textwidth]{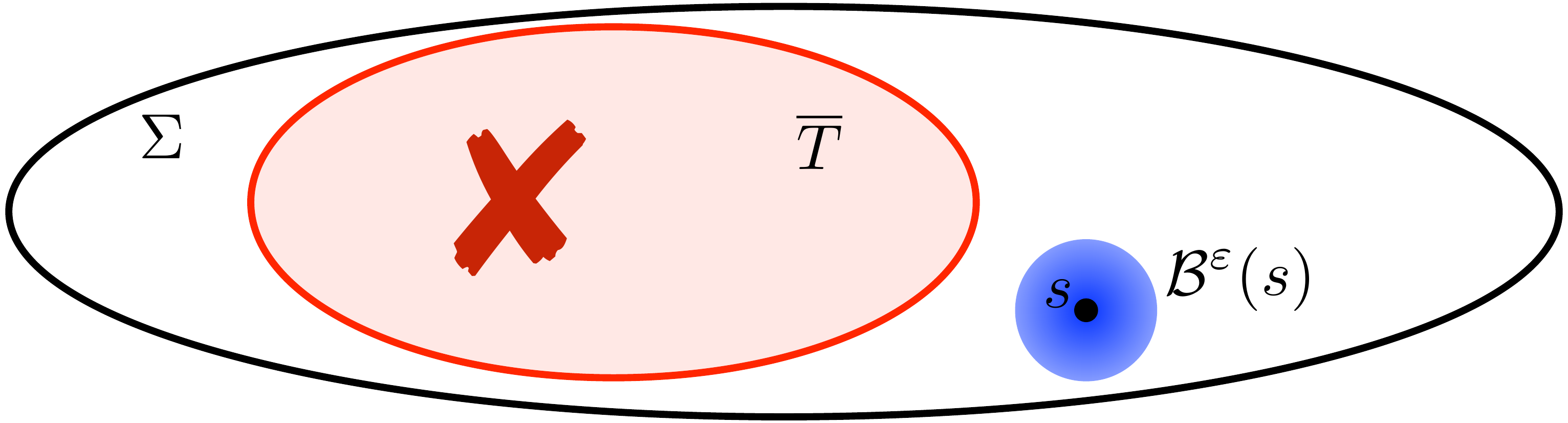}
\end{center}
\caption{\emph{Robustness.} A theory $T$ satisfies \Robust{} if for any story~$s\in \Sigma$ that is not forbidden by $T$ there exists a neighbourhood $\cB^{\varepsilon}(s)$ of stories that are  not forbidden, either.  This means that the set $\overline{T}$ of forbidden stories is closed.}
\label{fig_Robustness}
\end{figure}

\newcommand*{\hseprule}{ \arrayrulecolor{lightgray} \specialrule{.05em}{0.2em}{0.2em} }
\setlength{\tabcolsep}{6.8pt}
\begin{table*}[ht!]
{\small 
\begin{tabular}{cccc}
 \arrayrulecolor{black} \specialrule{0.1em}{0em}{0.3em}
   & \shortstack{story \\{}} &  \shortstack{plot $s^{\Exp}$  \\  $\Exp \in \ExpMeas_{\cH, \{\pi_{(z_1, \ldots, z_n)}\}}$} & \shortstack{plot $s^{\Exp^*}$   \\  $\Exp^* \in \ExpMeasRep_{\cK, \{\pi_z\}, f_{\hor} < 0.6}$} \\
 \arrayrulecolor{black} \specialrule{0.06em}{0.3em}{0.3em}
  $s_1$ & \texts{The source repeatedly emits photons with polarisation $\ket{\hor}$. The measurement outcomes are always $\hor$.} & $\bigl\{ (\ket{\hor}^{\otimes n}, [\hor \hor \hor \hor \cdots]_1^n :  n \in \mathbb{N} \bigr\}$ & $\bigl\{ (\ket{\hor}^{\otimes n}, n) :  n \in \mathbb{N} \bigr\}$ \\
   \hseprule
    $s_2$ & \texts{The source repeatedly emits photons with polarisation $\ket{\hor}$. The measurement outcomes alternate between $\ver$ and~$\hor$.} &  $\bigl\{ (\ket{\hor}^{\otimes n}, [\ver \hor \ver \hor \cdots ]_1^n :  n \in \mathbb{N} \bigr\}$ & $\bigl\{ (\ket{\hor}^{\otimes n}, \mathrm{ok}) :  n \in \mathbb{N} \bigr\}$  \\
   \hseprule
  $s_3$ & \texts{The source repeatedly emits photons with polarisation $\ket{\diag}$. The measurement outcomes are always~$\hor$.}  &  $\bigl\{ (\ket{\diag}^{\otimes n}, [\hor \hor \hor \hor \cdots]_1^n : n \in \mathbb{N} \bigr\}$ & $\bigl\{ (\ket{\diag}^{\otimes n}, n) : n \in \mathbb{N} \bigr\}$ \\
   \hseprule
    $s_4$ & \texts{The source repeatedly emits photons with polarisation $\ket{\diag}$. The measurement outcomes alternate between $\ver$ and~$\hor$.} &  $\bigl\{ (\ket{\diag}^{\otimes n}, [\ver \hor \ver \hor \cdots]_1^n  : n \in \mathbb{N} \bigr\}$  & $\bigl\{ (\ket{\diag}^{\otimes n}, \mathrm{ok}) : n \in \mathbb{N} \bigr\}$ \\
   \hseprule
    $s_5$ & \texts{The source repeatedly emits photons with polarisation $\ket{\diag}$. The number of measurement outcomes $\ver$ is at any time at least as large as the number of outcomes~$\hor$.} &  undefined & $\bigl\{ (\ket{\diag}^{\otimes n}, \mathrm{ok}) : n \in \mathbb{N} \bigr\}$ \\
   \hseprule
    $s_6$ & \texts{The source repeatedly emits photons with polarisation $\ket{\diag}$. The number of measurement outcomes $\hor$ is at any time at least twice as large as the number of outcomes~$\ver$.}  &  undefined  & $\bigl\{ (\ket{\diag}^{\otimes n}, n) : n \in \mathbb{N} \bigr\}$ \\
 \arrayrulecolor{black} \specialrule{0.1em}{0.3em}{0em}
\end{tabular}
}

\caption{\emph{Example stories and their plots.} The stories talk about a prepare-and-measure experiment as depicted by Fig.~\ref{fig_ExpLight}. We write $\ket{\ver}$ and $\ket{\hor}$ for the vertical and horizontal polarisation, respectively, and $\ket{\diag} =
  (\ket{\hor}  +  \ket{\ver}) / \! \sqrt{2}$ for one of the diagonal directions. The precise meanings of the stories is specified by their plots. The third column shows the plots for the case where the event space is taken to be that of~\eqref{eq_PMEventspace}. Its elements specify the joint prepared state and the sequence of outcomes~$z = (z_1, \ldots, z_n)$.  Stories~$s_5$ and $s_6$ are not precise enough to define these. The right column shows the plots for the variant of the experiment that is used to formulate \BornFreq{}, whose event space is defined by~\eqref{eq_PMERepeventspace}. Here the output sequence $(z_1, \ldots, z_n)$ is replaced by the result  $t$ of a test; $t=\text{``$\mathrm{ok}$''}$ if the frequency $f_{\hor}$ of results~$z_i = \hor$ is below an upper threshold, $\theta = 0.6$, and $t=n$ otherwise.\label{tab_StoriesMeas}}
\end{table*}

\subsection*{Postulate~\BornDet{}} 

Postulate~\BornDet{} asserts that an outcome~$z$  whose projector $\pi_z$ has no overlap with~$\psi$ cannot occur. Within our framework, the postulate is phrased as a property of a theory~$T$.

\begin{nobreak}
\begin{emphbox}
  $T$ satisfies~\BornDet{} if it forbids any story~$s$ according to which the implication
  \begin{align*}
     \pi_z \psi =  0 \implies (\psi, z) \notin s^{\Exp} 
   \end{align*}
   is violated for some $\Exp \in \ExpMeas_{\cH, \{\pi_z\}}$.
\end{emphbox}
\end{nobreak}

Note that this condition is pretty weak, for the implication is only violated if $\psi$ is exactly orthogonal to $\pi_z$.  It rules out none of the examples of Table~\ref{tab_StoriesMeas}.

\subsection*{Postulate~\Robust{}}

Our second postulate, \Robust{}, makes criteria like the above tolerant to perturbations. It demands from a theory~$T$ that any story that is not forbidden by~$T$ has a neighbourhood of non-forbidden stories (cf.\ Fig.~\ref{fig_Robustness}). 

To turn this into a precise statement, we use that the event space $[\Exp]$ of any prepare-and-measure experiment  $\Exp \in \ExpMeas_{\cH, \{\pi_z\}}$ has a natural metric~$d$  induced by the inner product $\spr{\cdot}{\cdot}$ of the Hilbert space~$\cH$, i.e., 
\begin{align*}
  d\bigl((\psi_1, z_1), (\psi_2, z_2)\bigr) = \begin{cases}  \sqrt{\spr{\psi_1\!-\!\psi_2}{\psi_1\!-\!\psi_2}} & \text{if $z_1 = z_2$} \\ \infty & \text{if $z_1 \neq z_2$.} \end{cases}
\end{align*}
Since plots are subsets of the event space, we need to turn $d$ into a metric for sets of events. The canonical way to do this is to use the corresponding Hausdorff distance, which we denote by~$D$.\footnote{The \emph{Hausdorff distance} $D$ between two subsets $S$ and $S'$ of a metric space $(M, d)$ is defined as $D(S, S') = \inf \{\varepsilon: {S \subset\cB^{\varepsilon}(S')} \, \wedge \, {S' \subset \cB^{\varepsilon}(S)}\}$, where $\cB^{\varepsilon}(\cdot)$ denotes the $\varepsilon$-ball around its argument set, i.e., $\cB^{\varepsilon}(S) = \bigcup_{x \in S} \{{y \in M} : \, {d(x, y)} \leq \varepsilon\}$.}  The metric~$D$ can then be pulled back to the set of stories~$\Sigma$ using the (partial) function $s \mapsto s^{\Exp}$. That is, we define $D^{\Exp}(s_1, s_2) = D(s_1^{\Exp}, s_2^{\Exp})$, with the convention that $D^{\Exp}(s_1, s_2) = \infty$ for $s_1 \neq s_2$ whenever $s_1^{\Exp}$ or $s_2^{\Exp}$ is undefined. In this way, any experiment~$\Exp$ gives rise to a canonical metric, $D^{\Exp}$, on the set $\Sigma$ of stories.  

\begin{emphbox}
  $T$ satisfies \Robust{} if the set $\overline{T}$ of stories it forbids is closed w.r.t.\ the topology induced by~$D^{\Exp}$, for any $\Exp \in \ExpMeas_{\cH, \{\pi_z\}}$.
\end{emphbox}

\subsection*{Postulate \BornFreq{}}

Before stating the precise definition of \BornFreq{}, let us first have a look at the familiar example of a (classical) coin tossing experiment. If we repeatedly throw a fair coin, we would expect the relative frequency $f_{\mathrm{heads}}$ of outcome ``heads'' to get closer and closer to $\frac{1}{2}$, but to keep fluctuating around that value. To make a more definitive statement, we may fix a threshold $\theta > \frac{1}{2}$ and, after any toss, check whether  $f_{\mathrm{heads}}$ is below $\theta$. Suppose that in each round when this is not the case, i.e.,  whenever $f_{\mathrm{heads}} \geq \theta$, we get a yellow card. If $\theta$ is only slightly larger than $\frac{1}{2}$, we will probably get many yellow cards at the beginning. However, they will become more and more rare, until we will stop getting any. Intuitively, a condition for a coin to be fair is hence that, for any threshold $\theta > \frac{1}{2}$, we will only get a finite number of yellow cards. 

Applying this idea to quantum mechanics, we may consider a prepare-and-measure experiment where many identical subsystems with Hilbert space $\cK$ are prepared and measured individually with respect to a family of projectors $\{\pi_{z_i}\}_{{z_i} \in \cZ}$ on $\cK$. Just like in the case of coin tossing, one can think of the preparation and measurement of the subsystems $\cK$ as a sequential process (cf.\ Fig.~\ref{fig_ExpLight} for an example). Suppose that the experiment is equipped with a test device that continuously checks whether the relative frequency $f_{\zspec}$ of a predefined outcome $\zspec \in \cZ$ in the measured tuple $(z_1, \ldots, z_n)$  is below a fixed threshold, $\theta$. Whenever this is the case, it outputs $t=\text{``$\mathrm{ok}$''}$, and else $t=n$, corresponding to a yellow card with the number~$n$ of the current round written on it. (This number will make the counting of yellow cards easier.) The test device thus effectively carries out a measurement on the overall state space~$\cH$, which has the form~\eqref{eq_Fock}, with respect to the family of projectors $\{\Pi_t\}_{t \in \{\mathrm{ok}\} \cup \mathbb{N}}$ defined by
\begin{align} \label{eq_projectorform}
  \Pi_n = \sum_{\substack{(z_1, \ldots, z_n) \in \cZ^{\times n} \\|\{i: \, z_i = \zspec\}| \geq \theta n}} \bigotimes_{i=1}^n  \pi_{z_i} 
\end{align}
for $t = n \in \mathbb{N}$, and $\Pi_{\mathrm{ok}} = \id_{\cH}-\bigoplus_{n=1}^\infty \Pi_n$ for $t = \mathrm{ok}$. We denote the set of all such experiments, which is a subclass of $\ExpMeas_{\cH, \{\Pi_t\}}$, by~$\ExpMeasRep_{\cK, \{\pi_z\}, f_\zspec < \theta}$. Any experiment $\Exp^*$ of this class  has an event space of the form
\begin{align} \label{eq_PMERepeventspace}
  [\Exp^*] =  \{ (\Psi, t) : \, \Psi \in \cH, \, t \in \{\mathrm{ok}\} \cup \mathbb{N} \}  \ ,
\end{align}
where $\Psi$ is the joint state of all subsystems and $t$ is  the outcome of the test  (cf.\ Table~\ref{tab_StoriesMeas} for example plots). 

Postulate~\BornFreq{} refers to the particular case where all subsystems are  prepared in the same state $\psi \in \cK$, so that their joint state, $\Psi \in \cH$, lies in the subspace
\begin{align*}
  \cH_{\psi} = \mathrm{span} \{\psi^{\otimes n}: \, n \in \mathbb{N}_0\} \ .
\end{align*} 
The postulate demands that, for any given threshold $\theta > \| \pi_\zspec \psi \|^2$,  if one repeats the measurement $\{\pi_{z_i}\}_{z_i \in \cZ}$ sufficiently often, the relative frequency $f_{\zspec}$ of outcome $\zspec$ will remain below~$\theta$. In other words, for any story~$s$, there is an upper bound on the number~$n \in \mathbb{N}$ for which a yellow card is issued according to~$s$.

\begin{emphbox}
 $T$ satisfies \BornFreq{} if it forbids any story~$s$ according to which the implication
  \begin{align*}
    \| \pi_\zspec \psi \|^2 <  \theta 
    \implies
    \bigl|\!\{n : \exists \Psi \in \cH_\psi,  (\Psi, n) \in s^{\Exp^*}\}\!\bigr| < \infty
 \end{align*}
 is violated for some $\Exp^* \in \ExpMeasRep_{\cK, \{\pi_z\}, f_{\zspec} < \theta}$.
 \end{emphbox}

This criterion obviously rules out stories $s_2$, $s_3$, and $s_6$ of Table~\ref{tab_StoriesMeas}. Considering~$s_4$, one may also regard the creation and measurement of two photons  as one single repetition of a prepare-and-measure experiment, with prepared state~$\ket{\diag}^{\otimes 2}$ and  a measurement with respect to $\{{\pi_{\ver} \otimes \pi_{\ver}}, {\pi_{\ver} \otimes \pi_{\hor}}, {\pi_{\hor} \otimes \pi_{\ver}}, {\pi_{\hor} \otimes \pi_{\hor}}\}$. It is then straightforward to see that the criterion also rules out~$s_4$. 

\subsection*{Claim~\eqref{eq_claimBornFreq}}

Everything is now in place to state the first main result.

\begin{emphbox}
\vspace{-1.5ex}
\begin{theorem}  \label{thm_BornFreq}
 If a theory $T$ satisfies \BornDet{} and \Robust{} then it also satisfies \BornFreq{}. 
\end{theorem}
\end{emphbox}

\begin{proof}
  Consider any experiment $\Exp^*$ from the set $\ExpMeasRep_{\cK, \{\pi_z\}, f_\zspec < \theta}$, with $\theta \in [0,1]$, as well as any $\psi \in \cK$ such that $\| \pi_\zspec \psi\|^2 < \theta$. We need to show that any story~$s$ that violates the implication in the definition of~\BornFreq{}, i.e., for which
\begin{align} \label{eq_infinityn}
  \bigl|\{n : \exists \Psi \in \cH_\psi,  (\Psi, n) \in s^{\Exp^*}\}\bigr| = \infty
 \end{align}
holds,   is forbidden by any theory~$T$ that satisfies \BornDet{} and \Robust{}.
  
  For any $m \in \mathbb{N}_0$, let
\begin{align*}
  F_{\theta}^{\geq m} = \biggl( \bigoplus_{n=0}^{m-1} \id_{\cK}^{\otimes n} \biggr) \oplus \biggl( \bigoplus_{n=m}^{\infty} \left( \id_{\cK}^{\otimes n}- \Pi_n \right) \biggr) \ ,\end{align*}
 with $\Pi_n$ (which depends on $\theta$) defined by~\eqref{eq_projectorform}, be the projector onto the subspace of~$\cH$ associated to all outcome tuples $z=(z_1, \ldots, z_n)$ except those that have at least length~$m$ and whose relative frequency of entries $z_i=\zspec$ is at least $\theta$.  In particular, we have 
\begin{align} \label{eq_Fcut}
  n  \geq m \quad \implies \quad \Pi_n F_{\theta}^{\geq m} = 0 \ .
\end{align}
We first argue that applying the projector $F_{\theta}^{\geq m}$ onto any state $\Psi \in \cH_{\psi}$ leaves that state almost unchanged when $m$ is large. For this we use that a state of the form $\psi^{\otimes n}$ lies, for $n$ large, almost entirely in the (typical) subspace generated by projectors $\pi_{z_1} \otimes \cdots \otimes \pi_{z_n}$ for tuples $(z_1, \ldots, z_n)$  whose relative frequency of $\zspec$ is close to $\|\pi_\zspec \psi\|^2$. A quantitative variant of this statement is  Lemma~\ref{lem_distanceconv} in the appendix, with $\pi_0 = \pi_\zspec$ and $\pi_1$ the sum of all $\pi_z$ with $z \neq \zspec$. Writing $\Psi = \sum_n \alpha_n \psi^{\otimes n}$, where $\alpha_n$ are coefficients such that $\sum_{n} |\alpha_n|^2 = 1$,  we obtain 
\begin{align*}
  \bra{\Psi} F_{\theta}^{\geq m} \ket{\Psi} 
  & = \sum_{n \in \mathbb{N}_0} |\alpha_n|^2 \bra{\psi^{\otimes n}} F_{\theta}^{\geq m} \ket{\psi^{\otimes n}} \\
  &  = 1- \sum_{n \geq m} |\alpha_n|^2 \bra{\psi^{\otimes n}} \Pi_n \ket{\psi^{\otimes n}}  \\
  & \geq 1 - \max_{n \geq m} n \, e^{-2n(f - \| \pi_\zspec \psi\|^2)^2} \ .
\end{align*}
The expression in the maximum over $n$ tends to $0$ for $n$ large. We have thus established that
\begin{align} \label{eq_stateconv}
   \lim_{m \to \infty} \sup_{\Psi \in \cH_\psi} \bigl\| F_{\theta}^{\geq m} \Psi - \Psi \bigr\| = 0 \ .
\end{align}

Next we use $F_{\theta}^{\geq m}$ to modify stories. Let $\cF_{\theta, \psi}^{\geq m}$ be the function on the event space $[\Exp^*]$ defined by
\begin{align*}
  \cF_{\theta, \psi}^{\geq m}(\Psi, t) = \begin{cases}
    (\frac{F_{\theta}^{\geq m} \Psi}{\|F_{\theta}^{\geq m} \Psi\|}, t) & \text{if $\Psi \in \cH_{\psi}$} \\
    (\Psi, t) & \text{otherwise.}
    \end{cases}
\end{align*}
Applying this function element-wise to the events of a plot induces a function on the set $\Sigma$ of stories, which we also denote by $\cF_{\theta, \psi}^{\geq m}$. That is, for any $s \in \Sigma$,  $\cF_{\theta, \psi}^{\geq m}(s)$ is a story identical to~$s$, except that its plot for experiment~$\Exp^*$ is modified to
\begin{align*}
  \cF_{\theta, \psi}^{\geq m}(s)^{\Exp^*} = \{ \cF_{\theta, \psi}^{\geq m}(\Psi, t) : \, (\Psi, t) \in s^{\Exp^*}\} \ .
\end{align*}
It follows from~\eqref{eq_stateconv} that
\begin{align} \label{eq_storyconv}
 \lim _{m \to \infty} D^{\Exp^*}\!(\cF_{\theta, \psi}^{\geq m}(s),  s) = 0 \ .
\end{align}

Consider now any theory~$T$ that satisfies \BornDet{} and \Robust{}, and let $s$ be any story for which~\eqref{eq_infinityn} holds. The latter implies that for any $m \in \mathbb{N}$ there exists $n \geq m$ and $\Psi \in \cH_\psi$ such that 
\begin{align*}
   \left({\textstyle \frac{F_{\theta}^{\geq m} \Psi}{\|F_{\theta}^{\geq m} \Psi\|}, n} \right)  \in \cF_{\theta, \psi}^{\geq m}(s)^{\Exp^*} \ .
\end{align*}
According to~\eqref{eq_Fcut}, the projector $\Pi_n$ has no overlap with the modified state on the left hand side.  \BornDet{} thus implies that  the modified story is forbidden, i.e., 
\begin{align*}
 \cF_{\theta, \psi}^{\geq m}(s) \in \overline{T}  \ ,
\end{align*}
for any $m \in \mathbb{N}$. In addition, we know from~\eqref{eq_storyconv} that the modified stories, for $m$ large,  approximate~$s$, i.e., 
\begin{align*}
  \lim_{m \to \infty} \cF_{\theta, \psi}^{\geq m}(s) = s \ .
\end{align*}
Since $\overline{T}$ is  closed due to \Robust{}, we can conclude that $\overline{T}$ must also contain~$s$. We have thus established that $s$ is forbidden by~$T$, which is what we set out to prove. 
\end{proof}

Before concluding this section, we note that the converse of Theorem~\ref{thm_BornFreq} may not hold. That is, there could exist a theory $T$ that satisfies  \BornFreq{}, but violates \mbox{\BornDet{}} or \Robust{}.  To see this, consider a story~$s$  that is defined like~$s_4$ in Table~\ref{tab_StoriesMeas}, but asserts that the sequence of individual outcomes $z_i$, instead of alternating between $\hor$ and $\ver$, corresponds to a binary representation of the number~Pi. Since Pi has a finite description, $s$ is a valid story. Furthermore, since the binary representation of Pi has no (known) repeating pattern, a theory $T$ can satisfy \BornFreq{} without forbidding~$s$. Conversely, following the lines of the proof above, one can construct a sequence of stories that approximate $s$ and yet violate the condition in \mbox{\BornDet{}}.  Hence, if $T$ satisfies~\BornDet{} and~\Robust{} then $T$ must forbid~$s$, too. This establishes that \BornDet{} and \Robust{} together are more restrictive than \BornFreq{}.

A similar argument also shows why it is important to demand that $\Sigma$ be countable. If it was not, $\Sigma$ could contain stories like $s$ above, but now for all possible infinite sequences of outcomes. A theory that satisfies \BornDet{} and \Robust{} would then necessarily rule out all of them.

\section{Reintroducing Probabilities} \label{sec_Bayes}

In this section we explain and prove claim~\eqref{eq_claimBornBayes},  thus establishing a Bayesian reading of the Born rule. We start by providing definitions for  \Repeat{}, \Symmetry{}, and \BornBayes{}. Unlike \BornDet{} and \Robust{}, which we  defined in the previous section, they should not be regarded as \emph{(objective) physical postulates}. Rather, they are attributes of an \emph{agent's subjective belief} about future observations. They also do not depend on quantum mechanics. To phrase them, we may consider any data-generating process, provided that  it is in principle repeatable.

According to the Bayesian approach, an agent's belief is represented by probabilities~\cite{deFin89,deFin37}. Consider, for instance, a quantity~$Z$ that is only revealed later to the agent. We may then assign to any possible value  $z$ of $Z$  a real number, $P(z)$, with  the following meaning:
\begin{align*}
  P(z)  \cong \texts{maximum amount that the agent would be ready to pay for a bet with payoff $1$ if $Z=z$ and $0$ otherwise}
\end{align*}
One can show that, if the agent's reasoning is rational, in the sense that no combination of her bets would lead to a sure loss, then the values $P(z)$ satisfy the usual axioms of probability theory~\cite{Ramsey31}.  In particular,  they are non-negative and sum up to~$1$. 

Suppose now that the agent enters the following gambling game. She is given a randomly chosen initial bonus $M \in \mathbb{N}$,\footnote{We assume that $P_{M}(m) > 0$ for any $m \in \mathbb{N}$, which ensures that conditioning on the event $M = m$  is well defined.} and then plays in rounds, numbered by~$n$. In each of them, the same process is invoked to generate a data point~$Z_n$ with alphabet~$\cZ$. At the start of each round, the agent must pay an entry fee of~$1$, but then earns the amount~$r$ if (and only if) a particular outcome, $Z_n=\zspec$, occurs, where $\zspec \in \cZ$ and $r>1$ are constants. The rules also say that the agent must play at least $M$ rounds (which is always possible with the initial bonus),  but has to stop when she runs out of money. Denoting by~$\mathbf{Z}$ the tuple consisting of the outcomes $Z_n$ obtained during the entire game and by $|\mathbf{Z}|$ its length, corresponding to the number of rounds played, we have
\begin{align} \label{eq_totalnumberbound}
  M \leq |\mathbf{Z}| \leq M + r \sum_{n=1}^{|\mathbf{Z}|-1} \delta(Z_n, \zspec) \ ,
\end{align}
where $\delta(x,y) = 1$ if $x=y$ and $\delta(x,y)=0$ otherwise. 

The definitions below all refer to an agent's belief about the quantities $\mathbf{Z}$ and $M$ occurring in this game.  Technically, they are expressed as properties of the joint probability distribution $P_{\mathbf{Z} M}$.

\subsection*{Property~\Repeat{}} 

Property~\Repeat{}, which enters as an assumption in claim~\eqref{eq_claimBornBayes},  captures the idea that the agent's personal belief about the outcome $Z_n$ of the next round does not depend on whether it is planned to continue later.  In the gambling game, the value $M$ has no other effect than setting lower and upper bounds on  the number of repetitions, which we stated in~\eqref{eq_totalnumberbound}. As a necessary requirement, we may therefore demand that the gambler's belief about $Z_n$ before starting round~$n$ is independent of~$M$. 

\begin{emphbox}
\addtocounter{footnote}{1} \renewcommand{\thempfootnote}{\thefootnote} 
  An agent's belief satisfies~\Repeat{} if\footnotemark
  \begin{align*}
P_{Z_n | Z_1 \cdots Z_{n-1}, |\mathbf{Z}| \geq n}  = P_{Z_n | Z_1 \cdots Z_{n-1}, |\mathbf{Z}| \geq n, M = m}
\end{align*}
for any $n, m \in \mathbb{N}$. 
 \end{emphbox}
 
 \addtocounter{footnote}{-1}
 \footnotetext{The equality is meant to hold for all arguments $(z_1, \ldots, z_n)$ for which the conditional probability $P_{Z_n | Z_1 \cdots Z_{n-1}, |\mathbf{Z}| \geq n}$ is defined, which is the case whenever $P_{Z_1 \cdots Z_{n-1} | |\mathbf{Z}| \geq n}(z_1, \ldots, z_{n-1})$ is strictly positive.}
 
 Note that the condition $|\mathbf{Z}| \geq n$ ensures that the outcome $Z_n$ is defined.

\subsection*{Property~\Symmetry{}} 

This property appears as another assumption in claim~\eqref{eq_claimBornBayes}. It demands that, if the game is played for at least $m$ rounds, then the agent's belief about the first $m$ outcomes, $Z_1, \ldots, Z_m$,  does not depend on their ordering. 

\begin{emphbox}
  An agent's belief satisfies~\Symmetry{} if 
  \begin{align*}
    P_{Z_1 \cdots Z_{m} | M = m} = P_{Z_{\pi(1)} \cdots Z_{\pi(m)} | M =m}
  \end{align*}
  for any $m \in \mathbb{N}$ and any permutation $\pi$ on $\{1, \ldots, m\}$.
 \end{emphbox}

\subsection*{Property~\BornBayes{}} 

So far, we have not said anything about the process that generates the values $Z_n$  in the gambling game. But now, to phrase property~\BornBayes{} | the Bayesian interpretation of the Born rule | we obviously need to bring in quantum mechanics. We therefore assume that the data-generating process is of the prepare-and-measure type as described earlier. For \BornBayes{} it suffices to consider one single round of the game. However, Theorem~\ref{thm_BornBayes} below uses \Repeat{} and \Symmetry{} as assumptions, which involve multiple rounds, and one should think of each round as consisting of the same state preparation and measurement. 

\begin{emphbox}
  An agent's belief satisfies~\BornBayes{} if 
  \begin{align*}
      P_{Z_1}(\zspec) =  \| \pi_\zspec \psi \|^2   \ ,
  \end{align*}
  for any $\zspec \!  \in \! \cZ$, whenever the prepared state is $\psi \in  \cK$, and the measurement is carried out with respect to $\{\pi_z\}_{z \in \cZ}$.
 \end{emphbox}

\subsection*{Claim~\eqref{eq_claimBornBayes}}

Theorem~\ref{thm_BornBayes} below is a verbose formulation of claim~\eqref{eq_claimBornBayes}. In addition to \Repeat{} and \Symmetry{},  which were explained above, it is assumed that the agent's personal belief is \emph{compatible with} a given theory~$T$.  By this we mean that the agent assigns probability~$0$ to all events that only occur according to stories $s \in \Sigma$ that are forbidden by~$T$.

\begin{emphbox}
  \vspace{-1.5ex}
\begin{theorem}  \label{thm_BornBayes}
 If an agent's belief satisfies \Repeat{} and \Symmetry, and is compatible with a theory for which~ \BornFreq{} holds, then it also satisfies \BornBayes{}. 
\end{theorem}
\end{emphbox}

\begin{proof}
We will show that, under the assumptions of the theorem, for any $\zspec \in \cZ$,
  \begin{align} \label{eq_BornBayesIneq}
  P_{Z_1}(\zspec) \leq \| \pi_\zspec \psi \|^2
\end{align}
holds.  \BornBayes{}, which has an equality instead of an inequality, then follows because, when taking the sum over $\zspec \in \cZ$, both sides of~\eqref{eq_BornBayesIneq} add up to~$1$.

Suppose that the agent plays the gambling game described above, with the constant~$r$ set to any value within the open interval $(1, \| \pi_\zspec \psi\|^{-2})$, where we assume without loss of generality that $\| \pi_\zspec \psi\|^2 < 1$ (otherwise~\eqref{eq_BornBayesIneq} is trivial). We first argue that, if the agent decides to continue playing for as long as possible, she will  run out of money with certainty.
    
Let $\theta$ be any fixed real from the open interval  $(\| \pi_\zspec \psi\|^2, r^{-1})$. Furthermore, for any $n, m \in \mathbb{N}$, let
\begin{align*}
  \Theta^n_{m} = \bigl\{(z_1, \ldots, z_n): \, \inf_{k \leq n}  m - k + r \sum_{i=1}^{k} \delta(z_i, \zspec) \geq 1 \bigr\}
\end{align*}
be the set of all tuples $\mathbf{z} = (z_1, \ldots, z_n)$ of  outcomes of the first $n$ rounds for which an initial bonus of $M=m$ coins suffices to play for at least one more round~$n+1$.  For $n$ sufficiently larger than $m$, namely
  \begin{align*}
  n \geq \frac{m}{r^{-1}-\theta} \ ,
  \end{align*}
 any $\mathbf{z} \in \Theta^n_{m}$ satisfies
\[
   \sum_{i=1}^n \delta(z_i, \zspec)
   \geq \frac{n-m}{r}
   \geq  \theta n  \ ,
\]
which is the property measured by the projector $\Pi_n$ defined by \eqref{eq_projectorform}. However, since $\theta > \| \pi_\zspec \psi\|^2$, stories according to which this property holds for $n \in \mathbb{N}$ arbitrarily large are forbidden by any theory that fulfils \BornFreq{}. We have thus established that the agent, if her strategy was to continue playing as long as she has money left, must  assign probability~$1$ to the event that she will run out of money after finitely many rounds. Note also that the probability distributions occurring in the assumptions \Repeat{} and \Symmetry{} do not depend on the choice of strategy. We can therefore assume in the following that the agent decides to play until she runs out of money (although this is obviously not a profitable strategy). We thus get a finite  tuple~$\mathbf{Z}$ of outcomes such that
\begin{align} \label{eq_haltcondition}
  (Z_1, \ldots, Z_{n-1}) \in \Theta^{n-1}_{M} \iff | \mathbf{Z} | \geq n 
\end{align}
holds for any $n \in \mathbb{N}$.

We now define a distribution $P^{*}_{\bar{\mathbf{Z}}}$ over (infinite) sequences $\bar{\mathbf{Z}} = (\bar{Z}_n)_{n \in \mathbb{N}}$  by 
\begin{align*}
  P^{*}_{\bar{Z}_n | \bar{Z}_1 \cdots \bar{Z}_{n-1}} 
  & = P_{Z_n | Z_1 \cdots Z_{n-1}, M \geq n} \\
  & =  P_{Z_n | Z_1 \cdots Z_{n-1}, |\mathbf{Z}| \geq n, M \geq n}  \ ,
\end{align*}
where the second equality holds because $M \geq n \implies |\mathbf{Z}| \geq n$ (see~\eqref{eq_totalnumberbound}), ensuring that the probabilities for $Z_n$ are defined. In the following we write $\mathbf{z}_1^n$ for the first $n$ entries of a tuple $\mathbf{z}$. Because of \Repeat{} and~\eqref{eq_haltcondition}, we have
\begin{align} \label{eq_Pstarm}
P^{*}_{\bar{Z}_n | \mathbf{\bar{Z}}_1^{n-1} = \mathbf{z}_1^{n-1}} 
& = P_{Z_n | \mathbf{Z}_1^{n-1} = \mathbf{z}_1^{n-1}, M = m} 
\end{align}
for any $n, m \in \mathbb{N}$ and any tuple $\mathbf{z}$ satisfying $\mathbf{z}_1^{n-1} \in \Theta^{n-1}_{m}$. By induction over~$n$, and using that 
\begin{align*}
  \mathbf{z}_1^{n} \in \Theta^n_{m} \implies \mathbf{z}_1^{n-1} \in \Theta^{n-1}_{m} \ ,
\end{align*}
we can turn~\eqref{eq_Pstarm} into
\begin{align} \label{eq_starequiv}
  P^{*}_{\bar{Z}_1 \cdots \bar{Z}_n}(\mathbf{z}_1^{n}) = P_{Z_1 \cdots Z_n | M = m}(\mathbf{z}_1^{n})
\end{align} 
whenever $\mathbf{z}_1^{n-1} \in \Theta^{n-1}_{m}$. This also implies that
\begin{align}  \label{eq_PstarPeq}
  P^{*}_{\bar{Z}_1 \cdots \bar{Z}_n} = P_{Z_1 \cdots Z_n | M = m}
\end{align} 
whenever $n \leq m$, because under this condition the set $\Theta^{n-1}_{m}$ contains all possible $(n-1)$-tuples.  
  
Condition~\eqref{eq_haltcondition}, together with the fact (established above) that $\mathbf{Z}$ has a finite length $|\mathbf{Z}| \geq m$ with certainty, implies that, conditioned on any $M=m$, the event
\begin{align*}
  \exists n \geq m: \, \mathbf{Z}_1^{n-1} \in \Theta^{n-1}_{m} \quad \text{and} \quad   \mathbf{Z}_1^{n} \notin \Theta^{n}_{m}
\end{align*}
occurs with certainty.  Inserting the definition of $\Theta^n_{m}$, this is equivalent to the claim that the event
\begin{align*}
  \exists n \geq m: \, \mathbf{Z}_1^{n-1} \in \Theta^{n-1}_{m} \quad \text{and} \quad m - n + r \sum_{i=1}^{n} \delta(Z_i, \zspec) < 1 
\end{align*}
occurs with certainty conditioned on $M = m$. Note that this statement only involves probabilities of $n$-tuples $\mathbf{Z} = (Z_1, \ldots, Z_n)$ such that $\mathbf{Z}_1^{n-1} \in \Theta^{n-1}_{m}$. Hence, by virtue of~\eqref{eq_starequiv}, we can conclude that the above event also occurs with certainty when $\mathbf{Z}$ is replaced by $\mathbf{\bar{Z}}$ sampled according to the probability distribution~$P^*_{\mathbf{\bar{Z}}}$. In particular, for any $m \in \mathbb{N}$, the event
\begin{align*}
  \exists n  \geq m: \, \mathbf{\bar{Z}}_1^{n-1} \in \Theta^{n-1}_{m} \quad \text{and} \quad \frac{1}{n} \sum_{i=1}^{n} \delta(\bar{Z}_i, \zspec) < r^{-1}
\end{align*}
occurs with certainty. We have thus established that
\begin{align} \label{eq_Pstarfreq}
  \liminf_{n \to \infty} \frac{1}{n} \sum_{i=1}^{n} \delta(\bar{Z}_i, \zspec) \leq r^{-1}
\end{align}
holds with certainty for $\bar{\mathbf{Z}} = (\bar{Z}_n)_{n \in \mathbb{N}}$ sampled according to $P^*_{\bar{\mathbf{Z}}}$.  

To conclude the argument, we recall that $r$ was an arbitrary real from the open interval $(1, \| \pi_\zspec \psi\|^{-2})$, i.e., 
\begin{align} \label{eq_Pstarfreqmu}
  \liminf_{n \to \infty} \frac{1}{n} \sum_{i=1}^{n} \delta(\bar{Z}_i, \zspec) \geq \| \pi_\zspec \psi\|^2 + \varepsilon
\end{align}
has probability~$0$ for any $\varepsilon > 0$.  Because of the assumption \Symmetry{} and~\eqref{eq_PstarPeq}, the sequence $\mathbf{\bar{Z}}$ is exchangeable. The representation theorem of de Finetti~\cite{deFin37} implies that such sequences have the property that their joint distribution can be written as a convex combination of distributions of the form $P_{\bar{Z}}^{\times n}$. Furthermore, this convex combination must contain distributions $P_{\bar{Z}}$ such that $P_{\bar{Z}}(\zspec) \geq P^*_{\bar{Z}_1}(\zspec)$. This implies that, with probability strictly larger than zero, the frequency of $\zspec$ in the sequence $\mathbf{\bar{Z}}$ is equal or larger than $P^*_{\bar{Z}_1}(\zspec)$. This statement is made precise by Lemma~\ref{lem_singlepfreq} in the appendix. Using now that~\eqref{eq_Pstarfreqmu} has probability~$0$, the lemma implies
\begin{align*}
P^*_{\bar{Z}_1}(\zspec) <  \| \pi_\zspec \psi\|^2  + \varepsilon \ .
\end{align*}
Using that  $\varepsilon > 0$ is arbitrary and, again, \eqref{eq_PstarPeq}, we get
\begin{align*}
  P_{Z_1 | M = m}(\zspec) \leq \| \pi_\zspec \psi\|^2
\end{align*}
for any $m \geq 1$, which immediately implies~\eqref{eq_BornBayesIneq}. 
\end{proof}

Combining Theorem~\ref{thm_BornBayes} with Theorem~\ref{thm_BornFreq}, we  obtain the following corollary. 

\begin{emphbox}
\vspace{-1.5ex}
\begin{corollary}  \label{cor_BornBayes}
 If an agent's belief satisfies \Repeat{} and \Symmetry{}, and is compatible with a theory for which \BornDet{} and \Robust{} hold,  then it also satisfies \BornBayes{}. 
\end{corollary}
\end{emphbox}

%\subsection{Further examples}
%
%\RR{Norton's dome is an excellent example to illustrate how randomness can be introduced into a classical theory. According to Norton's argument, the direction into which the ball rolls is not determined. But, obviously, classical mechanics does not assign any probability distribution to it. However, we can do it if we make an additional assumption on our knowledge. For example, it is natural to assume that, because the dome is rotational symmetric, an agent should have a prior that is invariant under rotations. In this case, we would find that the ball rolls into each direction with equal probability.}

\section{Conclusions}

Theorem~\ref{thm_BornFreq} and~\ref{thm_BornBayes} show that the experimentally relevant consequences of the Born rule can as well be obtained from two alternative physical postulates, \BornDet{} and \Robust{} | together with certain natural assumptions about rational reasoning. These postulates have interesting features, which make them suitable as potential substitutes for the Born rule. Firstly, they are more general, for they are not restricted to experiments where a particular measurement is repeated arbitrarily often under identical conditions.  Secondly, they come closer to the idea that physical postulates should be generic principles rather than specific quantitative statements. 
 
The argument presented here may also shed new light on the nature of the Born rule itself. Is it an (objective) physical law, as suggested by \BornFreq{}? Or is it rather a statement about (subjective) beliefs, as in \BornBayes{}? While the former view is implicit to many standard quantum mechanics textbooks, the latter  is probably most consequently advocated by QBism~\cite{FuchsSchack13}, which regards the Born rule as an empirical addition to  Bayesian probability theory.  Corollary~\ref{cor_BornBayes} suggests that the Born rule could in reality just be a ``blend'', consisting both of objective and subjective ingredients.
%\footnote{Maybe this is what Fuchs meant when he wrote in Footnote~14 of~\cite{Fuchs10} that ``QBism finds its happiest spot in an unflinching combination of `subjective probability' with `objective indeterminism'.'' Indeed,  \BornDet{} and \Robust{} only specify what outcomes can \emph{possibly} be observed when measuring a system, and therefore are indeterministic (but objective) laws.} 
Indeed, \BornDet{} and~\Robust{} are of the same objective kind as the usual physical laws, whereas \Repeat{} and~\Symmetry{} are manifestly subjective.  The Born rule may hence be viewed as the result of taking (objective) physical postulates and supplementing them with rules for (subjective) rational  reasoning.

\acknowledgments

We thank L\'idia del Rio, Artur Ekert, and R\"udiger Schack for discussions. This research was supported by the Swiss National Science Foundation (SNSF) via the NCCR ``QSIT'' and by the Air Force Office of Scientific Research (AFOSR) via grant FA9550-16-1-0245. 

\appendix

\section*{Appendix} \label{app_tech}

The following bound is used in the proof of Theorem~\ref{thm_BornFreq}. Results of this type are well established in information theory~\cite{Hoeffding63}. They are sometimes referred to as typicality bounds. For completeness, we nevertheless provide a statement with a proof.

\begin{lemma} \label{lem_distanceconv}
  Let $\pi_0, \pi_1$ be projectors on a Hilbert space $\cK$ such that $\pi_0 + \pi_1 = \id_{\cK}$, let $\psi \in \cK$ be normalised, let $k, n \in \mathbb{N}$, and define
  \[
    \Pi = \sum_{\substack{(b_1, \ldots, b_n) \in \{0,1\}^n \\ \sum_i b_i = n-k}} \pi_{b_1} \otimes \cdots \otimes \pi_{b_n} \ .
  \]
  Then
  \begin{align*}
    \bra{\psi^{\otimes n}} \Pi \ket{\psi^{\otimes n}} \leq e^{-2 n (\bra{\psi} \pi_0 \ket{\psi} - \frac{k}{n})^2 } \ .
  \end{align*}
\end{lemma}

\begin{proof}
  Let $p_b = \bra{\psi} \pi_b \ket{\psi}$ for $b \in \{0,1\}$, and note that
  \begin{align*}
     \bra{\psi^{\otimes n}} \Pi \ket{\psi^{\otimes n}} 
   = \binom{n}{k} p_0^{k} p_1^{n-k} \ .
   \end{align*}
   For $k=0$, the right hand side equals $p_1^n = e^{n \ln p_1}$, which cannot be larger than $e^{-2 n (1-p_1)^2}$, so that the claimed bound holds. The same is true, analogously, for $k=n$. For $k \in \{1, \ldots, n-1\}$ we use Stirling's approximation to bound the binomial, leading to
   \begin{align*}
      \bra{\psi^{\otimes n}} \Pi \ket{\psi^{\otimes n}}    & < \frac{e}{2 \pi} \sqrt{\frac{n}{k (n-k)}} \frac{n^n}{k^{k} (n-k)^{n-k}}   p_0^{k} p_1^{n-k} \ .
  \end{align*}
   With the definition $q_0 = \frac{k}{n}$ and $q_1 = \frac{n-k}{n}$, this can be further bounded by
  \begin{align*}
     \bra{\psi^{\otimes n}} \Pi \ket{\psi^{\otimes n}} 
   & < e^{-n(q_0 \ln \frac{q_0}{p_0} + q_1 \ln \frac{q_1}{p_1})} \\
   & \leq e^{-2 n (p_0 - q_0)^2} \ ,
  \end{align*}
  where we have used Pinsker's inequality~\cite{Csiszar67,Kulback67} in the second line.
\end{proof}

The lemma below, which is used in the proof of Theorem~\ref{thm_BornBayes}, is a relatively straightforward consequence of the de Finetti representation theorem~\cite{deFin37}. 

\begin{lemma} \label{lem_singlepfreq}
  Let $(X_i)_{i \in \mathbb{N}}$ be an exchangeable sequence of random variables on $\cX$, i.e., 
  \begin{align*}
    P_{X_1 \cdots X_n} = P_{X_{\pi(1)} \cdots X_{\pi(n)}} 
  \end{align*}
  for any $n \in \mathbb{N}$ and any permutation $\pi$ on $\{1, \ldots, n\}$.  Then, for any $\xi \in \cX$, the event 
  \begin{align*}
    \lim_{n \to \infty} \frac{1}{n} \sum_{i = 1}^n \delta(X_i, \xi) \geq P_{X_1}(\xi)  
  \end{align*}
  has non-zero probability. 
\end{lemma}

\begin{proof}
  Because of the exchangeability property, de Finetti's theorem asserts that there exists a probability measure $d\mu$ on the distributions $Q_X$ over $\cX$ such that
  \begin{align} \label{eq_deFin}
    P_{X_1 \cdots X_n} = \int Q_X^{\times n} d\mu(Q_X) 
  \end{align}
  for any $n \in \mathbb{N}$. We may therefore interpret $Q_X(\xi)$ as a random variable distributed according to $d \mu$ such that the distribution of the $n$-tuple $(X_1, \ldots, X_n)$ conditioned on $Q_X$ satisfies 
  \begin{align*}
    P_{X_1 \cdots X_n | Q_X} = Q_X^{\times n} \ .
  \end{align*}
   Hence, by the strong law of large numbers, the relative frequency of the symbol~$\xi$ in that $n$-tuple equals $Q_X(\xi)$, i.e., 
    \begin{align*}
      \lim_{n \to \infty} \frac{1}{n} \sum_{i = 1}^n \delta(X_i, \xi)    = Q_X(\xi)
  \end{align*}
  with certainty. It thus remains to show that $Q_X(\xi) \geq P_{X_1}(\xi)$ with non-zero probability.  To see this, assume by contradiction that the claim is wrong, i.e., that $Q_X(\xi) < P_{X_1}(\xi)$ with certainty. But then, taking the average over~$Q_X$, we obtain
  \begin{align*}
    \int Q_X(\xi) d\mu(Q_X)  
    < P_{X_1}(\xi) \ ,
   \end{align*}
   which contradicts~\eqref{eq_deFin}.
  \end{proof}

\bibliography{Foundations}

\end{document}